\documentclass{llncs}
\usepackage{amsmath,amstext,amssymb}

\usepackage{tikz}
\usepackage{graphics}
\usepackage{graphicx}
\usepackage{algorithm}
\usepackage[noend]{algorithmic}
\usetikzlibrary{arrows, automata}

\def\trace{{\it trace}}
\def\Var{{\it Var}}

\def\span{{\it Span}}
\newcommand{\UA}{{U^TAU}}
\newcommand{\ignore}[1]{}
\renewcommand{\Pr}{\mathop{\bf Pr\/}}
\newcommand{\E}{\mathop{\bf E\/}}
\newcommand{\Ex}{\mathop{\bf E\/}}

\newcommand{\R}{\mathbb R}

\newcommand{\Z}{{\mathbb Z}}

\newcommand{\eps}{\epsilon}

\newcommand{\half}{{\textstyle \frac12}}

\newcommand{\bits}{\{-1,1\}}

\newcommand{\ynote}[1]{}
\newcommand{\pnote}[1]{}
\newcommand{\knote}[1]{}

\begin{document}

\title{Optimal query complexity for estimating the trace of a matrix }
\author{Karl Wimmer\inst{1}\thanks{ Most of the work is done when the author is visiting the  Simons Institute for the Theory of Computing, University of California-Berkeley. Supported in part by NSF grant CCF-1117079. } \and Yi Wu\inst{2}\thanks{Most of the work is done when the author is visiting the Simons Institute for the Theory of Computing, University of California-Berkeley.} \and Peng Zhang\inst{2}}
\institute{Duquesne University  \and Purdue University}
\maketitle

\begin{abstract}
Given an implicit $n\times n$ matrix $A$ with oracle access $x^TA x$ for any $x\in \R^n$,  we study the query complexity of  randomized algorithms for estimating the \emph{trace} of the  matrix.  This problem has many applications in quantum physics, machine learning, and pattern matching. Two metrics are commonly used for evaluating the estimators: i) variance; ii) a high probability multiplicative-approximation guarantee.  Almost all the known estimators are of the form $\frac{1}{k}\sum_{i=1}^k x_i^T A x_i$ for $x_i\in \R^n$ being i.i.d. for some special distribution. 

Our main results are  summarized as follows:
\begin{enumerate}
\item We give an \emph{exact} characterization of the minimum variance unbiased estimator  in the broad class of  \emph{linear nonadaptive} estimators (which subsumes all the existing known estimators).
\item We also consider the query complexity lower bounds  for \emph{any} (possibly nonlinear and adaptive) estimators:
\begin{enumerate}
\item We show that \emph{any} estimator requires $\Omega(1/\eps)$ queries to have a guarantee of variance at most $\eps$.
\item We show that \emph{any} estimator requires $\Omega(\frac{1}{\eps^2}\log \frac{1}{\delta})$ queries to achieve a $(1\pm\eps)$-multiplicative approximation guarantee with probability at least $1 - \delta$.
\end{enumerate}
Both above lower bounds are asymptotically tight.
\end{enumerate}

As a corollary, we also resolve a conjecture in the seminal work of Avron and Toledo (Journal of the ACM 2011) regarding the sample complexity of the Gaussian Estimator.

\end{abstract}

\thispagestyle{empty}
\newpage
\setcounter{page}{1}

\section{Introduction}
Given an $n \times n$ matrix $A=\{A_{ij}\}_{1\leq i\leq n, 1\leq j\leq n}$,  we study the problem of estimating its trace
 \[
 \trace(A) = \sum_{i=1}^n A_{ii}
 \]
with a randomized algorithm that can query $f_A(x) = x^T A x$ for  any $x\in \R^n$. The goal is to minimize the number of queries used to achieve certain type of accuracy guarantee, such as the variance of the estimate or a multiplicative approximation (which holds with high probability). Finding an estimator that achieves such an accuracy guarantee with few queries has several applications. For example, this problem is  well studied in the subject of  lattice quantum chromodynamics, since such queries are physically feasible and can be used to efficiently estimate the trace of a function of a large matrix $f(A)$. Such an estimator can also be used as a building block for many other applications including solving least-squares problems~\cite{Hut89}, computing the number of triangles in a graph~\cite{Avr10,Tso08}, and string pattern matching~\cite{A01,AGW13}. 

This problem has been well studied in the literature.  All of the previously analyzed estimators are of the form $\frac{1}{k} \sum_{i=1}^k  x_i^T  A x_i$ for $x_1,x_2,\ldots, x_k\in \R^n$; nearly all take $x_1,x_2,\ldots,x_k$ to be independent and identically distributed (i.i.d.) from some well designed distribution. For example, in~\cite{Hut89}, the author just takes each query to be a \ignore{Rademacher} random vector whose entries are i.i.d. Rademacher random variables (i.e., each coordinate is a uniformly random sample from $\{-1,1\}$); we call this the Rademacher estimator.  There are also several other alternative distributions on $x_1, x_2, \ldots, x_k$, such as drawing each query from a multivariate normal distribution~\cite{SR97}, we call this the Gaussian estimator.  Here, the coordinates of each vectors are i.i.d. Gaussian random variables.  The work of~\cite{IE04} considers the case where only one query is allowed, but that query can be a unit vector in $\mathbb{C}^n$. Other estimators occur in~\cite{drabold1993maximum,wang1994calculating}.  Recent work by~\cite{AT11}, the authors propose several new estimators such as the unit vector estimator, normalized Rayleigh-quotient trace estimator, and the mixed unit vector estimator.  One estimator that does not use i.i.d. queries is due to~\cite{roosta2013improved}; in that work, the authors propose querying random standard basis vectors without replacement.

To characterize the performance of an estimator, perhaps the most natural metric is the variance of the estimator. It is known that the Gaussian estimator has variance $2 \|A\|^2_F$ and the random Rademacher vector estimator has variance $2( \|A\|^2_F-\sum_{i=1}^n A_{ii}^2)$, where $\|A\|_F = \sqrt{\trace(A^TA)}$ is the Frobenius norm.  In recent work by Avron and Toledo~\cite{AT11}, it is suggested that the notion of a multiplicative approximation guarantee might be a better success metric of an estimator than the variance. Formally, we say an estimator is an $(\eps,\delta)$-estimator if it outputs an estimate in the interval $\left((1-\eps)\trace(A) , (1+\eps)\trace(A)\right)$ with probability at least $1-\delta$.  It should be noted that some assumptions on the matrices need to be made to have a valid $(\eps,\delta)$-estimator, as it is impossible to achieve any multiplicative approximation when the matrix could have a trace of $0$.  A natural choice is to assume that $A$ comes from the class of symmetric positive semidefinite (SPD) matrices. For a SPD matrix, the authors in~\cite{AT11} prove that the Gaussian estimator with $k=O(\frac{1}{\eps^2} \log (\frac{1}{\delta}))$ queries to the oracle is an $(\eps,\delta)$-estimator.  It was recently shown in~\cite{roosta2013improved} that the random Rademacher vector estimator is also an $(\eps,\delta)$-estimator with the same sample complexity. 

An open problem asked in~\cite{AT11} is the following:  does the Gaussian estimator require $\Omega(\frac{1}{\eps^2} \log (\frac{1}{\delta}))$ in order to be an $(\eps,\delta)$-estimator?  The authors showed that this number of queries suffices and conjectured that their analysis of the Gaussian estimator is tight with supporting evidence from empirical experiments.  The paper gives some intuition on how to show an $\Omega(\frac{1}{\eps^2})$ lower bound.  The authors suggested that the difficulty of turning this argument into a formal proof is that ``current bounds [on the $\chi^2$ cumulative distribution function] are too complex to provide a useful lower bound''.   Regarding lower bounds for trace estimators, we note the related work of~\cite{LNW14}, which considers the problem of sketching the nuclear norms of $A$ using bilinear sketches (which can be viewed as nonadaptive queries of the form $x^T Ay$). The problem is similar to estimating trace when the underlying matrix is positive semidefinite.
  
All of the above mentioned estimators (with one exception in~\cite{roosta2013improved}) use independent identically distributed queries from some special distributions, and the output is a linear combination of the query results. On the other hand, when viewing an estimator as a randomized algorithm, we can choose any distribution over the queries, and the output can be any (possibly randomized) function of the results of the queries.  Given the success of the previously mentioned estimators, it is natural to ask whether these extensions are helpful.  For example, can we get a significantly better estimator with a non i.i.d. distribution?  Can we do better with adaptive queries?  Can we do better with a nonlinear combination of the query results?

In this paper, we make progress on answering above questions and understanding the  optimal query complexity for randomized trace estimators. Below is an informal summary of our results.
\begin{enumerate}
\item Among all the linear nonadaptive trace estimators (which subsumes all the existing trace estimators), we prove  that the ``random $k$ orthogonal vector'' estimator is the minimum variance estimator.  The distribution on the queries is \emph{not} i.i.d., and we are unable to find an occurrence of this estimator in the literature regarding trace estimators.
\item We also prove two asymptotically optimal lower bounds for any (possibly adaptive and possibly nonlinear) estimator.
\begin{enumerate}
\item We show that \emph{every} trace estimator requires $\Omega(1/\eps)$ queries to have a guarantee that the variance of the estimator is at most $\eps$.
\item We show that \emph{every} $(\eps,\delta)$-estimator requires $\Omega(\frac{1}{\eps^2}\log \frac{1}{\delta})$ queries.
\end{enumerate}
\end{enumerate}

As a simple corollary, our result also confirms the above mentioned conjecture in~\cite{AT11} (as well as the tightness of the analysis of the Rademacher estimator in~\cite{roosta2013improved}). Notice our result is a much stronger statement:  the original conjectured lower bound is only for an estimator that returns a linear combination of i.i.d. Gaussian queries; we prove the lower bound holds for any estimator.  Our lower bound also suggests that  adaptiveness as well as nonlinearity will not help asymptotically as all these lower bounds are matched by the nonadaptive Gaussian estimator.  On the other hand, our upper bound suggests that the exact minimum variance estimator might not use i.i.d. queries.

\subsection{Problem Definitions}
\begin{definition}[estimator for the trace] \label{def:set} A trace estimator is a  randomized algorithm that, given query access to an oracle $f_A(\cdot)$ for an unknown $n\times n$ matrix $A$, makes a sequence of $k$ queries $x_1, x_2,\ldots, x_k \in \R^n$ to the oracle and receives $f_A(x_1), f_A(x_2), \ldots, f_A(x_k)$.  The output of the estimator is a real number $h(A)$ determined by the queries and the answers to the queries. 
\end{definition}

\begin{definition}[nonadaptive linear unbiased trace estimator] We say a trace estimator is \emph{nonadaptive} if the distribution of $x_i$ is not dependent on $f_A(x_1), f_A(x_2),\ldots, f_A(x_{i-1})$. A trace estimator is \emph{linear} if we sample from a distribution over $k$ queries as well as their weights: $(x_1,x_2,\ldots,x_k)$, and $(w_1,w_2,\ldots, w_k)$, and output $\sum w_i f_A(x_i)$.  In addition, a linear trace estimator is \emph{unbiased} if
\[
         \E_{w_1,w_2, \ldots, w_n,x_1,x_2,\ldots, x_n}[\sum_{i=1}^n w_i f_A(x_i)] = \trace(A)
\]
\end{definition}

Without loss of generality, we can assume that all the queries in a linear estimator are of unit length, where the actual lengths of the queries are absorbed by the weights.

The most natural measure of quality of an estimator is its variance.
There is a large body of work on the existence of and finding a \emph{minimum variance unbiased estimator}.  Such an estimator has a strong guarantee; it is the estimator for which the variance is minimized for all possible values of the parameter to estimate.  In general, finding such an estimator is quite difficult.  It is easy to see that the variance depends on the scale of the matrix.  To normalize, we assume that the Frobenius norm of the matrix is fixed.

\begin{definition}
\label{def:var-estimator}
We define the variance of a trace estimator as the worst case of variance over all matrices with Frobenius norm $1$. To be specific, given a matrix $A$\ let us define $\Var(A,h)=\Ex[(h(A)-\trace(A))^2]$, then

\[
\Var(h)=\sup_{ \|A\|_F^2 = 1 } {\Var(A,h)}.
\]

If the variance of an estimator $h$ is at most $\delta$, we say that $h$ is a $\delta$-variance estimator.
\end{definition}

Given an unbiased estimator class, the \emph{minimum variance unbiased estimator} has the minimum variance among all the (unbiased) estimators in the class.

Another natural accuracy guarantee for a trace estimator is the notion of $(\eps,\delta)$-estimator that is introduced in~\cite{AT11}.

\begin{definition}[$(\eps,\delta)$-estimator]
\label{def:eps-delta-estimator}
A trace estimator $h$ is said to be an $(\epsilon,\delta)$-estimator of the trace if, for every matrix $A$, we have that $ |\trace(A)-  h(A)| \leq \eps \cdot \trace(A)$ with probability at least $1-\delta$.
\end{definition}  

We stress that both Definitions~\ref{def:var-estimator} and~\ref{def:eps-delta-estimator} involve worst case estimates over the choice of the matrix, and the randomness only comes from the internal randomness of the estimator.

%
%

\subsection{Main Results}
Our main results are as follows:

\begin{theorem}\label{thm:main1} Among all linear nonadaptive unbiased trace estimators, the minimum variance unbiased estimator that makes $k$ queries is achieved by sampling $k$ random orthogonal unit vectors  (see Definition~\ref{def:transformation}) $x_1,x_2,\ldots, x_k$ and outputting $\frac{n}{k}\sum_i f_A(x_i)$.
\end{theorem}

\begin{theorem}\label{thm:main2} Any trace estimator with variance $\eps$ requires $\Omega(1/\eps)$ queries.  
\end{theorem}
\begin{theorem}\label{thm:main3} Any $(\eps,\delta)$-estimator for the trace requires $\Omega(\frac{1}{\eps^2}\log (\frac{1}{\delta}))$ queries, even if the unknown matrix is known to be positive semidefinite.
\end{theorem}
The bounds in~Theorem~\ref{thm:main2} and~\ref{thm:main3} are tight: both bounds can be asymptotically matched by the Gaussian estimator and the uniform Rademacher vector estimator.

\subsection{Proof Techniques Overview}
All of our results crucially use a powerful yet simple trick, which we call \emph{symmetrization}.  The heart of this trick lies in the fact that the trace of a matrix is unchanged under similarity transformations; $\trace(A) = \trace(U^TAU)$ for every $A$ and orthogonal $U$.  
If we have a nonadaptive estimator with query distribution $(x_1, x_2, \ldots,x_k) \sim P$ and an orthogonal matrix $U$, using the queries distributed as $(Ux_1, Ux_2, \ldots, Ux_k)$ should not be too different
in terms of worst-case behavior.  (We have to be more careful with adaptive estimators, which we discuss in Section~\ref{sec:sym}.)  Thus, applying symmetrization to a nonadaptive estimator yields a nonadaptive estimator where it draws queries as in the original estimator, but transforms the queries using a random orthogonal transformation.  This ``symmetrizes'' the estimator.  We prove that the performance of the estimator never decreases when symmetrization is applied, so we can exclusively consider symmetrized estimators.

In order to characterize the minimum variance linear nonadaptive unbiased estimator, we notice that after the symmetrization, the distribution over queries for any such estimator is defined by a distribution over the pairwise angles of the $k$ queries.  We then show that the queries should be taken to be orthogonal with certainty in order to minimize variance.  

As for the lower bounds for adaptive and nonlinear estimators, the symmetrization also plays an important role.  Consider the problem of  proving a query lower bound for $(\eps,\delta)$-approximation:  the most common approach of proving such a lower bound is to use Yao's minimax principle.  To apply this principle, we would need to construct two distributions of matrices such that the distributions cannot be distinguished after making a number of queries, even though the traces  of the matrices are very different in the two distributions.  There are several technical difficulties in applying the minimax principle directly here.  First of all, the query space is $\R^n$, so it is unclear whether one can assert that there exists a sufficiently generalized minimax principle to handle this case.  Second, even if one can apply a suitable version of minimax principle, we do not have general techniques of analyzing the distribution of $k$ adaptive queries, especially when the queries involve real numbers and thus the algorithm might have infinitely many branches.

\ignore{, which is characterized by a decision tree on real numbers of possibly infinite branches. (It will be much easier if we only consider lower bounds for nonadaptive query which is essentially characterized by a $k$-vector distribution).
}

We overcome the above two barriers and avoid using a minimax principle entirely by applying symmetrization.  One nice property of the symmetrization process is that a symmetrized estimator outputs the same distribution of results on all matrices with the same diagonalization. In the proof we carefully construct two distributions of matrices with the same diagonalization in each distribution, while the traces are different for different distributions.  Each distribution is simply the ``orbit'' of a single diagonal matrix $D$; the support consists of all matrices similar to $D$. Using the symmetrization, it suffices to show that we can not distinguish these two distributions of matrices by $k$ adaptive queries, as it is equivalent to distinguish two diagonal matrices for symmetrized trace estimators. The argument for achieving a lower bound for adaptive estimators is more subtle; we show that due to the structure of symmetrized estimators, we define a stronger query model such that adaptive estimators behave the same as the nonadaptive estimators while we achieve the same lower bound, even with the stronger query model.

\subsection{Organization}
In section~\ref{sec:pre}, we define the mathematical tools that are needed in our analysis.
In section~\ref{sec:sum}, we introduce the idea of symmetrization. We prove Theorem~\ref{thm:main1} in section~\ref{sec:linear}. In section~\ref{sec:varlb}, we prove Theorem~\ref{thm:main2}. In section~\ref{sec:approxlb}, we prove Theorem~\ref{thm:main3}. 


\section{Preliminaries} \label{sec:pre}
\begin{definition}[random Gaussian matrix and random orthogonal matrix]\label{def:transformation} 
\begin{itemize}
\item We call a vector $g\in \R^n$ a random Gaussian vector if each coordinate is sampled independently from $N(0,1)$.
\item We call an $n\times n$  matrix $G$ a random Gaussian matrix if its entries are sampled independently from $N(0,1)$.
\item We call an $n\times n$ matrix $U$ a random orthogonal matrix if it is drawn from the distribution whose probability measure is the Haar measure on the group of orthogonal matrices; specifically, it is the unique probability measure that is invariant under orthogonal transformations.

\item We call $k$ vectors $x_1, x_2,\ldots, x_k\in \R^n$  $k$  random orthogonal unit
vectors if they are chosen as $k$ row vectors of a random orthogonal matrix.
\end{itemize}
\end{definition}

We note that one way to generate a random orthogonal matrix is to generate a random Gaussian matrix and perform Gram-Schmidt orthonormalization on its rows.

%

\begin{definition}[total variation  distance] \label{def:vd} Let $P,Q$ be two distributions
with density functions $p,q$ over a domain $\Omega$.  The total variation distance
between $P$ and $Q$ is defined as $d_{TV}(P,Q) = \frac{1}{2}\int_{z\in \Omega}|p(z)-q(z)| \; dz$.
\end{definition}

\begin{proposition}\label{claim:vd} Suppose we are given one sample on either distribution $Q_1$ or $Q_2$ defined on the same sample space, and we are asked to distinguish which distribution the sample came from.  The success probability of any algorithm is at most $\half+\half d_{TV}(Q_1,Q_2)$.
\end{proposition}

\begin{definition}[KL-divergence] \label{def:kl} Let $P,Q$ be two distributions
with density functions $p,q$ over a domain $\Omega$.  The Kullback-Leibler divergence of $Q$ from $P$ is defined as $d_{KL}(P,Q) = \int_{z\in \Omega}\ln(p(z)/q(z)) p(z)dz$.
\end{definition}

We know the following relationship between Kullback-Leibler divergence and total variation distance, which is also known as Pinsker's Inequality.

\begin{theorem} \label{thm:kl} Suppose the KL-divergence between distributions $P,Q$ is $d_{KL}(P,Q)$, and total variation distance is $d_{TV}(P,Q)$, then
\[
2d_{TV}(P,Q)^2\leq  d_{KL}(P,Q)
\]
\end{theorem}

To bound the total variation distance of the distributions we consider, we compute the KL-divergence and apply Pinsker's Inequality.  For example, the KL-divergence between two multivariate Gaussian distributions is well known; we will only use the following special case for multivariate Gaussian distributions with mean $\vec{0}$ and covariance matrices $\Sigma_0$ and $\Sigma_1$.

\begin{theorem}\label{thm:klg} Given two $n$-dimensional Gaussian distributions $N_0 = N(\vec{0}, \Sigma_0)$ and  $N_1 = N(\vec{0}, \Sigma_1)$ we have that

\[
d_{KL}({N}_0 ,{N}_1) = { 1 \over 2 } \left( \trace \left(\Sigma_1^{-1} \Sigma_0 \right)  - n -\ln { \det(  \Sigma_0 ) \over \det( \Sigma_1 ) } \right)
\]

\end{theorem}

\ignore{
We will also need to bound the total variation distance between random orthogonal vectors and a multivariate Gaussian distribution.  We will apply two such results.

\begin{theorem}~\cite{Kho06} Given a random unit vector $v$ and a random unit Gaussian vector $g$, let us use $v_{[k]}$ and $g_{[k]}$ to denote the first $k$ coordinates of $v$ and $g$.
We have
\[
        d_{TV}(g_{[k]},v_{[k]}) = O(k/n)
\]
\end{theorem}

\begin{theorem} ~\cite{LNW14}\label{thm:LNW} Given a $n$ dimensional random orthogonal matrix $U$ as well as a random Gaussian matrix, if we take a $r \times k$ submatrix $U_{r,k}$ of $U$ and a $r \times k$ submatrix $G_{r,k}$ of $G$. We have that
\[
        d_{TV}(G_{r,k},U_{r,k})=o(1)
\]
when $r k\leq n^{1-\Omega(1)}$.

\end{theorem}
}

\section{Symmetrization of an estimator}\label{sec:sum}\label{sec:sym}
\ignore{
In this section, we introduce the idea of symmetrization of an estimator which is a crucial element of all our remaining proofs.

\begin{definition}[Rotation of an estimator]
\label{def:rotation-estimator}
For any $n\times n$ orthogonal matrix $U$ (i.e., $U^T U=I$) and estimator $h$, let us define $h^U$ as the estimator that, on input $A$, makes the queries that $h$ makes and receives responses that $h$ receives on input $U^TAU$ up to a transformation $U$; i.e., assuming $h$ queries $x_1,x_2, \ldots, x_k$ on $A'=U^T A U$, $h^U$ will query $U x_1, Ux_2, \ldots, Ux_k$ on $A$.
\end{definition}
}

In this section, we introduce the idea of symmetrization of an estimator which is a crucial element of all our remaining proofs.  We first define the \emph{rotation} of an estimator, which we will denote $h^U$ for an $n \times n$ orthogonal matrix $U$.  
Intuitively, the construction of $h^U$ is such that $h^U$ emulates the behavior of $h$ on a rotated version of the matrix $A$.  More specifically, $h^U$ makes queries in the following way:

\begin{itemize}
\item Letting $q_1$ be a random variable whose distribution is the same as the first query of $h$, the distribution of the first query of $h^U$ is the same as the random variable $Uq_1$.

\item Given queries $Uq_1,Uq_2,\ldots,Uq_{j-1}$ made by $h^U$ so far with responses $t_1,t_2,\ldots,t_{j-1}$, the distribution of the $j$th query of $h^U$ has the same distribution as $Uq_j$, where $q_j$ is distributed the same as the $j$th query that $h$ makes, given queries $q_1,q_2,\ldots,q_{j-1}$ with responses $t_1,t_2,\ldots,t_{j-1}$. 
\end{itemize}

In the case that $h$ is a nonadaptive estimator, the queries of $h^U$ are just $Ux_1,Ux_2, \ldots, Ux_k$, where $x_1,x_2,\ldots,x_k$ is a set of queries from the distribution of queries that $h$ makes.

\begin{lemma}\label{lem:rotation}For any estimator $h$ and orthogonal matrix $U$,
\begin{itemize}
\item $\Var(h^U)=\Var(h)$.
\item  $h$ is an $(\eps,\delta)$-approximation estimator if and only if $h^{U}$ is also an $(\epsilon,\delta)$-estimator.
\end{itemize}
\end{lemma}
\begin{proof}We know that given a matrix $A$, the behavior of $h^U$ is the same as $h$
on estimating $U^TAU$. On the other hand, we know that $\trace(U^T
A U)=\trace(A)$ and $\|A\|_F=\|\UA\|_F$. Therefore, the variance of $h^U$ on
$A$ is the same as the variance of $h$ on $\UA$.  Now suppose $h$ is an $(\eps,\delta)$-estimator.  We know that the approximation guarantee of $h^U$ on $A$ is the same as $h$ on $\UA$. Therefore, we know that with probability at least $(1-\delta)$, the estimator $h^U$'s output is within
\[
\left((1-\eps)\trace(\UA), (1+\eps)\trace(\UA)\right)=\left((1-\eps)\trace(A), (1+\eps)\trace(A)\right).
\]

\end{proof}

\begin{definition}[averaging estimators over a distribution]Suppose we have a collection of estimators $H$, for any probability distribution $P$ on $H$, we define $h_{H,P}$ as the following estimator:
\begin{enumerate}
\item Randomly sample an estimator $h\sim P$.
\item Output according to the  estimation of $h$.
\end{enumerate}
\end{definition}
\begin{lemma}
\label{lem:averaging}

Averaging a collection of estimators cannot increase variance or weaken an $(\eps,\delta)$-guarantee.  Specifically:

\begin{itemize}

\item If all the estimators $H$ are unbiased and have variance at most  $c$, then $h_{H,P}$'s variance is also at most $c$.

\item If all the  estimators in $H$ are $(\eps,\delta)$-estimators, then $h_{H,P}$ is also an $(\eps,\delta)$-estimator.
\end{itemize}
\end{lemma}

\begin{proof}
For the first, we apply the law of total variance conditioned on the draw of $h \sim P$:

\[
\Var[h_{H,P}] = \E_{h \sim P}[\Var[h]] + \mathop\Var_{h \sim P}[\E[h]]
\]

The second term above is $0$, since all estimators in $H$ are unbiased.
Since $\Var[h] \leq c$ for every $h \in H$, $\E_{h \sim P}[\Var[h]] \leq c$ as well.

For the second claim, assuming that 

\[
\Pr[h(A) \in \left((1 - \eps)\trace(A),(1 + \eps)\trace(A)\right) ] \geq 1 - \delta
\]

for each $h \in H$, we have

\[
\Pr[h_{H,P}(A) \in \left((1 - \eps)\trace(A),(1 + \eps)\trace(A)\right) ] \geq
\]
\[ 
\inf_{h \in H} \Pr[h(A) \in \left((1 - \eps)\trace(A),(1 + \eps)\trace(A)\right) ] \geq 1 - \delta.
\]

\end{proof}

\begin{definition}[Symmetrization of a trace estimator]Given an estimator $h$, we define the symmetrization $h^{sym}$ of $h$ to be the estimator where we
\begin{enumerate}
\item sample a random orthogonal matrix $U$ (see definition~\ref{def:transformation}), and
\item use $h^U$ to estimate the trace.
\end{enumerate}
We say an estimator is \emph{symmetric} if it is equivalent to the symmetrization of some estimator.
\end{definition}
By Lemma~\ref{lem:rotation} and Lemma~\ref{lem:averaging}, we know that $h^{sym}$'s variance as well as its $(\eps,\delta)$-approximation
is always no worse than $h$. Therefore, without loss of generality, we can always assume that the optimal estimator is symmetric.

One nice property of the symmetric estimator is that it has the same performance on all matrices with the same diagonalization.
\begin{lemma}
Given a symmetrized estimator $h^{sym}$, its variance and approximation guarantee is the same for any matrix $A$ and  $U^TA U$ for any orthogonal matrix $U$.
\end{lemma}
\begin{proof}
Given any matrix $A$, we know that the variance of $h^{sym}$ is $\E_{U_1}[\Var(h,U_1^TAU_1)]$ and the variance of $h^{sym}$ on the matrix $\UA$ is $\Ex_{U_{1}}[\Var(h, (U_{1}U)^{T}AU_1U)]$.  We know that $U_1$ and $U_1U$ are identically distributed; therefore, $h^{sym}$ has the same estimation variance on $A$ and $\UA$.

Similarly, for the approximation guarantee, suppose $h^{sym}$ is an $(\eps,\delta)$-estimator, which means that
\[
        \E_{U_1}\left[\Pr\left(h(U_1^TAU_1)\in \left((1-\eps)\trace(A),(1+\eps)trace(A)\right)\right)\right]\geq 1-\delta.
\]
We know that for the matrix $A'=UAU$, $U_1^TA'U_1=U_1^T\UA U_1$ has the same distribution as $U_1^TA U_1$. Therefore,
\begin{multline*}
\E_{U_1}\left[\Pr\left(h(U_1^TAU_1)\in \left((1-\eps)\trace(A),(1+\eps)trace(A)\right)\right)\right]
\\=\E_{U_1}\left[\Pr\left(h(U_1^TA'U_1)\in \left((1-\eps)\trace(A'),(1+\eps)trace(A')\right)\right)\right]
\end{multline*}
\end{proof}

\section{Optimal Linear Nonadaptive Estimator}\label{sec:linear}

Without loss of generality, we can assume that the optimal estimator is symmetric.  For a symmetric nonadaptive estimator,  we can think of $x_1, x_2, \ldots, x_k$ as generated by the following process. 

%
%

\begin{itemize} 
\item Sample a configuration $\theta= \{\theta_{ij}\}_{1\leq i < j \leq k}$ from some distribution $P_\Theta$ . For each configuration $\theta$, there is a corresponding weight vector $w^\theta=(w^{\theta}_1,w^{\theta}_2, \ldots, w^{\theta}_k)$.
\item Generate $x_1,x_2,\ldots, x_k$ by drawing $k$ random unit vectors conditioned on the angle between $x_i,x_j$ being $\theta_{ij}$ for all $i < j$.  (This can be done efficiently.)  
\item Output $\sum_{i=1}^k w^\theta_i  f_A(x_i)$.
\end{itemize}

The proof of Theorem~\ref{thm:main1} consists of two steps.  First we will show that we can set all of the angles (deterministically) to be $\frac{\pi}{2}$ without increasing the variance, so we can assume that the queries are orthogonal.  In the second step, we will then show that the optimal way of assigning weight is to (deterministically) set each weight to be $\frac{n}{k}$.

We first prove that we can replace the queries $x_1, x_2, \ldots, x_k$ by $k$ random orthogonal unit vectors without increasing the variance.
\begin{lemma}\label{lem:ort} Let $y_1,y_2, \ldots, y_k$ be $k$ randomly orthogonal unit vectors.  We have that 
\begin{equation}\label{eqn:ot}
        \Var\left(\sum_{i=1}^k w^\theta_i  f_A(y_i)\right)\leq \Var\left(\sum_{i=1}^k w^\theta_i  f_A(x_i)\right)
\end{equation}
\ignore{Over the class of unbiased nonadaptive linear estimators, the minimum variance is achieved by the estimator that queries a uniformly random set of $k$ mutually orthogonal unit vectors $x_1,\ldots,x_k$ and outputs $\sum_{i=1}^k \frac{n}{k} f_A(x_i)$.}
\end{lemma}
\begin{proof}

It is easy to see that the marginal distribution on each $x_i$ is the same as the marginal distribution on $y_i$.  Therefore, we have that $$\E_{\theta, y_1,\ldots,y_k}\left[\sum_{i=1}^k w^\theta_i  f_A(y_i)\right] =  \E_{x_1,\ldots,x_k,\theta}\left[ \sum_{i=1}^k w^\theta_i  f_A(x_i)\right] = \trace(A)$$ This implies that $\sum_{i=1}^k w^\theta_i  f_A(y_i)$ is also an unbiased estimator.

Since both estimators have the same expectation, in order to show~\eqref{eqn:ot}, it suffices to prove that 
\begin{equation}\label{eqn:o2}
\E_{\theta,x_1,\ldots, x_n}\left[\left(\sum_{i=1}^k w^\theta_i  f_A(x_i)\right)^2\right]\geq \E_{\theta,y_1,\ldots,y_n}\left[\left(\sum_{i=1}^k w^\theta_i  f_A(y_i)\right)^2\right]
\end{equation}
By the process of generating $x_1,x_2,\ldots, x_k$, we know that the marginal distribution of $x_i$ is independent of $\theta$ and equal to the marginal distribution of $y_i$. 
If we expand the left hand side of \eqref{eqn:o2}, we have that 
\begin{eqnarray*}
&&\E_{\theta,x_1,\ldots,x_n}\left[\left(\sum_{i=1}^k w^\theta_i  f_A(x_i)\right)^2\right]  \\
&=& \sum_{i=1}^k \E_{\theta}[(w^\theta_i)^2]\Ex_{x_i} [f_A(x_i)^2] +2 \sum_{1\leq i <j\leq k} \E_{\theta, x_i,x_j}[w_i^\theta w_j^\theta f_A(x_i)f_A(x_j)] 
\\& =&  \sum_{i=1}^k \E_{\theta}[(w^\theta_i)^2 ]\E_{y_i}[f_A(y_i)^2] + 2\sum_{1\leq i <j\leq k} \E_{\theta,x_i,x_j}[w_i^\theta w_j^\theta f_A(x_i)f_A(x_j)]
\end{eqnarray*}
If we expand the right hand side of \eqref{eqn:o2}
we have that 
\begin{eqnarray*}
&& \E_{\theta,y_1,\ldots, y_n}\left[\left(\sum_{i=1}^k w^\theta_i  f_A(y_i)^2\right)\right]  \\&=& \sum_{i=1}^k \E_{\theta}[(w^\theta_i)^2 ]\E_{y_i}[f_A(y_i)^2] +2 \sum_{1\leq i <j\leq k} \E_{\theta}[w_i^\theta w_j^\theta]\E_{y_i,y_j}[f_A(y_i)f_A(y_j)] 
\end{eqnarray*}
Therefore, in order to prove~\eqref{eqn:o2}, it suffices to prove that for any $i$ and $j$, we have
\begin{equation}\label{eqn:ij}
 \E_{\theta,x_i,x_j}[w_i^\theta w_j^\theta f_A(x_i)f_A(x_j)] \geq \E_{\theta}[w_i^\theta w_j^\theta] \E_{y_i,y_j}[f_A(y_i)f_A(y_j)]
\end{equation}

To compare $\E_{\theta,x_i,x_j}[w_i^\theta w_j^\theta f_A(x_i)f_A(x_j)]$ and $\E_{\theta}[w_i^\theta
w_j^\theta] \E_{y_i,y_j}[f_A(y_i)f_A(y_j)]$, we note that the marginal distribution on the pair $(x_i,x_j)$ is equivalent to drawing $x_i,x_j$ from the following process:
\begin{enumerate}
\item Draw $\theta\sim P_{\Theta}$.
\item Set $x_i = y_i$ and $x_j = y_i \cos \theta_{ij} + y_j \sin \theta_{ij}$.
\end{enumerate}
It is easy to check that the joint distribution on $x_i$ and $x_j$ has the same distribution as two random unit vectors with angle $\theta_{ij}$.

Therefore, 
\begin{eqnarray}\nonumber
 && \E_{\theta,x_i,x_j} [w_i^\theta w_j^\theta f_A(x_i)f_A(x_j)]  
 \\ \nonumber &=& \E_{\theta,y_i,y_j} [w_i^\theta w_j^\theta y_i^T A y_i (\cos \theta_{ij} \cdot y_i +\sin \theta_{ij} \cdot y_j)^{T}A(\cos \theta_{ij} \cdot y_i +\sin \theta_{ij} \cdot y_j)] 
 \\ \nonumber &=&
\E_{\theta} [w_i^\theta w_j^\theta  \cos ^2 \theta_{ij}]\E_{y_i} [ y_i^T A y_i \cdot y_i^T A y_i] + \E_{\theta} [w_i^\theta w_j^\theta  \sin ^2 \theta_{ij}]\E_{y_i,y_j} [ y_i^{T} A y_i  y_j^T A y_j] 
\\ \label{eqn:long} 
&+& \E_{\theta} [w_i^\theta w_j^\theta  \sin \theta_{ij}\cos \theta_{ij}]\E_{y_i,y_j} [ y_i^T A y_i  y_i^TA y_j +  y_i^T A y_i  y_j^T A y_i] 
\end{eqnarray}

%
%
%
%
%
In order to simplify the above expression, we first claim that 
\[
\E_{y_i,y_j} [ y_i^T A y_i  y_i^TA y_j+ y_i^T A y_i  y_j^TA y_i]  =0.
\]
To see this, note that $y_j$ is a random unit vector orthogonal to $y_i$.  Conditioned on any fixed realization of $y_i$, the distribution on $y_j$ is symmetric about $\vec{0}$; $y_j$ has the same distribution as $-y_j$.

In addition, using Cauchy-Schwarz and the fact that $y_i$ and $y_j$ have the same distribution, we have that
\[
        \E_{y_i} [( y_i^T A y_i )^2]  =\sqrt{\E_{y_i} [( y_i^T A y_i )^2]\E_{y_j} [( y_j^T A y_j )^2]}\geq \E_{y_i,y_j} [ y^T_j A y_j \cdot y_i^{T} A y_{i}].
\]

Therefore, we have that 
\begin{eqnarray*}
\eqref{eqn:long}&\geq & \E_{\theta} [w_i^\theta w_j^\theta  \cos ^2 \theta]\E_{y_i,y_j} [ y_j^T A y_j \cdot y_i^T A y_i] + \E_{\theta} [w_i^\theta w_j^\theta  \sin ^2 \theta]\E_{y_i,y_j} [ y_i^T A y_i  y_j^T A y_j] \\&=& \E_{\theta}[w^\theta_i w^\theta_j ]\cdot \E_{y_i,y_j}[y_j^T A y_j \cdot y_i^T A y_i] 
\end{eqnarray*}
which proves \eqref{eqn:ij}, completing the proof of Lemma~\ref{lem:ort}.
\end{proof}

Now that we can assume that the queries are mutually orthogonal, we can view this as an estimator with randomized weights $w_i^\theta$ for $\theta\sim P_{\Theta}$.  Below we will use the random variable $w_i$ to denote $w_i^\theta$ as $\theta$ is independent from $y_1,y_2, \ldots, y_k$.  

\begin{lemma}
\label{lem:equal-weights}

Let $(y_1,\ldots,y_k)$ be $k$ random orthogonal unit vectors.  Then the estimator $h=\sum_{i=1}^k w_if_A(y_i)$  has minimum variance when $w_1 = w_2 = \cdots = w_k = n/k$.
\end{lemma}
\begin{proof}
First we must have  $\E[\sum_{i=1}^k w_i ]= n$ to make the estimator unbiased, since $\E[f_A(y_i)]=\trace(A)/n$. Also, 
\[
        \E\left[\left(\sum_{i=1}^k w_i \right)^2\right]=\E\left[\sum_{i=1}^k w_i^2\right] +2\cdot \E\left[\sum_{1\leq i<j\leq k} w_iw_j\right]  \geq n^2
\]

Minimizing the variance is equivalent to minimizing 

\begin{eqnarray*}
&& \E_{w,y}\left[\left(\sum_{i=1}^k w_if_A(y_i)\right)^2\right] \\
&=& \sum_{i=1}^k  \E[f_A(y_i)^2]\E[w_i^2] + 2\cdot \sum_{1\leq i < j\leq k} \E[f_A(y_i) f_A(y_j)]\cdot \E[w_i w_j]\\
&=& \E\left[\sum_{i=1}^k w_i^2\right]\E[f^2_A(y_1)] + \left(\E\left[\left(\sum_{i=1}^k w_i\right)^2\right]-\E\left[\sum_{i=1}^k w_i^2\right]\right)\E[f_A(y_1)
f_A(y_2)]\\
&\ge& \frac{n^2}{k}\E[f_A(y_1)^2] + \left(n^2-\frac{n^2}{k}\right)\E[f_A(y_1)
f_A(y_2)]\\
&=& \E\left[\left(\sum_{i=1}^k \frac{n}{k}f_A(y_i)\right)^2\right]\\
\end{eqnarray*}
Equality holds 
for $w_1=\cdots=w_k=n/k$, completing the proof.
\end{proof}

Combining Lemma~\ref{lem:ort} and Lemma~\ref{lem:equal-weights}, the minimum variance linear nonadaptive unbiased estimator making $k$ queries is $\sum_{i=1}^k \frac{n}{k} f_A(y_i)$, where $y_1,y_2,\ldots,y_k$ is a collection of random orthogonal unit vectors.  This completes the proof of Theorem~\ref{thm:main1}.

\ignore{
\begin{theorem}
Linear estimator $h=\sum_{i=1}^k w_if_A(x_i)$ has the minimum variance when $x_1,\ldots,x_k$ are orthogonal uniform vectors and $w_1=\cdots=w_k=1/k$.
\end{theorem}
}

\section{Every $\delta$-variance estimator requires  $\Omega(1/\delta)$-queries}\label{sec:varlb}

\begin{theorem}\label{thm:varlb} Any estimator with variance $\delta$ requires $\Omega(1/\delta)$ queries.
\end{theorem}

The main idea to show such a lower bound is to reduce the problem from getting a $\delta$-variance estimator to an easier problem of distinguishing two distributions.  Here, the unknown matrix is drawn from either distribution $P_1$ or distribution $P_2$, and our goal is to determine, with high probability, which distribution the sample came from.  Importantly, the trace of every matrix in the support of $P_1$ is far from the trace of every matrix in the support of $P_2$.

For the problem here, we define $P_1$ and $P_2$ as follows, with parameter $\eps=O(\sqrt{\delta})$ precisely specified later :
\begin{enumerate}
\item A draw from distribution $P_1$ is the matrix $A_1 =1/\sqrt{5} (u u^T + 2 vv^T)$, where $u,v$ are two random orthogonal unit vectors.
\item A draw from distribution $P_2$ is the matrix $A_2 =  1/C\left((1+2\eps) u u^T +(2-\eps)vv^T\right)$, where $u,v$ are two random orthogonal unit vectors, for $$C=\sqrt{(1+2\eps)^2+(2-\eps)^2}=\sqrt{5(1+\eps^2)}$$
\end{enumerate}

It is easy to check that we have $\|A_1\|_F=\|A_2\|_F=1$.  It will be more convenient to write $A_2$ as that

\[
A_2 =  \frac{1}{\sqrt{5(1+\eps^2)}}\left((1+2\eps) u u^T +(2-\eps)vv^T\right)=\frac{1}{\sqrt{5}}((1+\eps_1)uu^T+(2-\eps_2)vv^T)
\]
where we have set $\eps_1= (1+2\eps)/\sqrt{1+\eps^2}-1$ and $\eps_2=2-(2-\eps)/\sqrt{1+\eps^2}$.  We will also write $\eps_3 = \eps_1-\eps_2$.  For the rest of the proof, we do not explicitly write $\eps_i$ as an expression of $\eps$.  Instead, our proof only requires that $\eps_1,\eps_2\leq O(\eps)$ and  $\eps_3= \Omega(\eps)$.  It is straightforward to verify that $|\trace(A_1)-\trace(A_2)|=\eps_3=\Omega(\eps)$.

The proof of Theorem~\ref{thm:varlb}  consists of the following two lemmas.

\begin{lemma}\label{lem:rank2}  Suppose there is an $\eps_3^2/60$-variance estimator that makes $k$ queries.  Then we can distinguish $P_1$ from $P_2$ with probability at least $2/3$ using $k$ queries.
\end{lemma}

\begin{lemma}\label{lem:testlb} Any algorithm that can distinguish $P_1$ from $P_2$ with probability at least $2/3$ requires at least $\Omega(1/\eps^2)$ queries.
\end{lemma}

Combining Lemma~\ref{lem:testlb} and Lemma~\ref{lem:rank2}, we have that
for any estimator with variance $\delta=\eps_3^2/60$, we need at least $\Omega(1/\delta)$ queries.

\begin{proof}(Proof of Lemma~\ref{lem:rank2}). Without loss of generality, let us assume that our estimator $h$ is symmetric.  By assumption, we have $\E_{h,A_i}[\Var(h, A_i)]\leq \eps_3^2/60$ for $i=1,2$. The randomness comes from the distribution of $A_i$ as well as the estimator $h$.  We use the fact that a symmetric estimator has the same variance on any matrix with the same diagonalization.

Consider the following algorithm for  distinguishing $A_1$ and $A_2$: it uses $h$ to first get an estimation $h(A)$ on the trace of the unknown matrix $A$.  Since every matrix in the support of $P_1$ has trace $3/ \sqrt{5}$ while every matrix in the support of $P_2$ has trace $(3+\eps_3)/ \sqrt{5}$, the algorithm will output that $A$ comes from the distribution $P_1$ if $h(A) \leq (3+\frac{1}{2}\eps_3)/ \sqrt{5}$, and $P_2$ otherwise.  Below we will show the accuracy of the above algorithm is at least $2/3$.

\ignore{
Given a unknown matrix from $A_1$ or $A_2$, by an averaging argument, we know then on at most $1/3$ of the $output$, the square of the estimation  error's above $\eps_3^2/4$ as the overall variance is less than $\eps_3^2/12$ . Therefore, for at least $2/3$ of the output, the estimator error square is less than $\eps_3^2/4$; i.e., the output is with an accuracy $\eps_3/2$ of the true trace. Since the trace difference between $A_1$ and $A_2$ is at least $\eps_3$, we know that  the algorithm will successfully distinguish $P_1,P_2$ with probability at least $2/3$.}

We know that $h(A)$ is a random variable satisfying $\E[(h(A) - trace(A))^2] \leq \eps_3^2/60$, where the randomness is only over the internal randomness of $h$.  By Chebyshev's inequality, we have $\Pr[| h(A)-\trace(A)| \ge \eps_3/(2\sqrt{5})] \le 1/3$. Thus, with probability at least $2/3$, if $A$ is drawn from $P_1$, then $h(A)<(3+\frac{1}{2}\eps_3)/\sqrt{5}$, so the algorithm succeeds. Similarly, with probability at least $2/3$, if $A$ is drawn from $P_2$, then $h(A)>(3+\frac{1}{2}\eps_3)/\sqrt{5}$, so the algorithm succeeds in this case as well.

\end{proof}

\begin{proof}(Proof of Lemma~\ref{lem:testlb}) First let us define a general problem.
\begin{definition} [distinguishing rank $2$ matrices]\label{def:general} We are given a matrix that is from one of the following two classes of distributions.
\begin{enumerate}
\item Distribution $P_{\alpha_1,\beta_1}$: $A_1 = u_1 u_1^T +  v_1v_1^T$ for $u_1 = \sqrt{\alpha_1} u$ and $v_1 = \sqrt{\beta_1} v$ for $u,v$ to be two  uniformly random orthogonal unit vectors.\item Distribution $P_{\alpha_2,\beta_2}$: $A_2 =  u_2 u_2^T +  v_2v_2^T$ for $u_2 = \sqrt{\alpha_2} u$ and $v_2 = \sqrt{\beta_2} v$ for $u,v$ to be two uniformly random orthogonal unit vectors.
\end{enumerate}
\end{definition}
It is easy to see that our goal is to understand the sample complexity of distinguishing $P_1= P_{\frac{1}{\sqrt{5}},\frac{2}{\sqrt{5}}}$ from $P_2=P_{\frac{1+\eps_1}{\sqrt{5}}, \frac{2-\eps_2}{\sqrt{5}}}$.

Recall that $\eps_1,\eps_2=O(\eps)$. We will prove the testing problem of distinguishing $P_1$ from $P_2$ requires $\Omega(1/\eps^2)$ even with stronger queries.  We define a \emph{strong query} to be a query $x$ that, instead of returning $x^T A x = (u \cdot x)^2 + (v \cdot x)^2$, returns $u \cdot x$ and $v \cdot x$.  The vectors $u$ and $v$ are vectors used to construct $A$ from a draw of $P_1$ or $P_2$ in Definition~\ref{def:general}.  It suffices to show that we need at least $\Omega(1/\eps^2)$ strong queries to distinguish $P_1$ from $P_2$.

One interpretation of the information given by a strong query is that it gives us the projection of $u_j,v_j$ on the query point $x$. After making $i$ queries $x_1,x_2,\ldots,x_i$, we have the projection of both $u_j,v_j$ on $\span_i= \span(x_1,x_2,\ldots, x_i)$.

We claim that if we are allowed to make strong queries, then it suffices to only consider estimators whose queries are a set of random orthogonal unit vectors.  The main idea is that, after any set of queries and responses, the essential structure of the problem, given this information, doesn't change too much, so the optimal next query is easy to determine.

\begin{lemma}\label{lem:strong} Without loss of generality, we can assume that the $k$ strong queries are a set of $k$ random orthogonal unit vectors, to distinguish $P_{\alpha_1,\beta_1}$ from $P_{\alpha_2,\beta_2}$ for any $\alpha_1,\alpha_2,\beta_1,\beta_2\geq 0$.
\end{lemma}

\begin{proof}
The proof is by induction on $k$.  First, we prove that the first query can be assumed to be a random unit vector.  This follows from the fact that the distributions $P_1$ and $P_2$ are rotation-invariant.  Thus, if we query $Ux_1$ instead of $x_1$ for an orthogonal transformation $U$, the distribution over the result of a query does not change due to the distribution of $u$ and $v$.

\ynote{say something more??}
Suppose that we have made queries $x_1, x_2, \ldots, x_{i-1}$ which are sampled from $i-1$ random orthogonal unit vectors. First we claim that without loss of generality, we can assume that  the $i$th query is orthogonal to $\span_{i-1}= \span(x_1,x_2, \ldots, x_{i-1})$. If not, we can always set the $i$th query to be its projection on $\span_{i-1}^{\perp}$, the orthogonal complement to $\span_{i-1}$ in $\R^n$. By doing this, we have the same amount of information as the resulting $\span_i$ does not change. Recall that by allowing strong queries, after making $i-1$ queries, all the information we have is the projection of $u_j$ and $v_j$ onto $\span_{i-1}$.

Next, we will show that we can assume $x_i$ is a random unit vector in $\span_{i-1}^{\perp}$. To see this,  given the querying result of $x_1, x_2, \ldots, x_{i-1}$, we know that the projection of $u_j,v_j$ (for the  unknown $j\in [2]$)  on $\span_{i-1}$. A crucial observation is that conditioned on any $x_1,x_2, \ldots, x_{i-1}$ as well as the result of the queries, we know that the projection of $u_j,v_j$ (for $j\in [2]$) behaves like random vectors of length $\sqrt{\alpha_j-l_u^{i-1}},\sqrt{\beta_j-l_v^{i-1}}$ lies in $\span_{i-1}^{\perp}$ for $l_u^{i-1}$ and $l_v^{i-1}$ being the  square length of $u,v$'s projection on $\span_{i-1}$ as $u,v$ are all randomly oriented. We also know that $i$th query completely lies in $\span_{i-1}^\perp$. Therefore, the $i$th query can be assumed to be a uniformly random unit vector in $\span_{i-1}^\perp$ by the same argument that we prove that $x_1$ can be assumed to be a uniformly random unit vector in $\R^n$.  Thus, $x_1,x_2,\ldots,x_i$ can be assumed to have the same distribution as $i$ random orthogonal unit vectors, completing the inductive step.
\end{proof}

Denote the results of the queries $x_1,x_2,\ldots,x_k$ as $(u_j\cdot x_1, v_j\cdot x_1),(u_j\cdot x_2,v_j\cdot x_2)\ldots,(u_j \cdot x_k,v_j\cdot x_k)$ for the unknown $j \in [2]$.  Because both $u_j,v_j$ and all $x_i$ are randomly oriented (conditioning on the $x_i$'s being orthogonal), we can further assume that $x_i = e_i$ for $e_i$ being the standard basis vector, that is, the vector with all its coordinate being $0$ except the $i$th coordinate being $1$. This does not change the distribution of the $k$ query results. Therefore, we can think of the optimal $k$-strong query test as a test for distinguishing the following two distributions over pairs of vectors in $\R^k$:

\begin{enumerate}
\item The distribution $U_1$, where a draw from $U_1$ is the pair $(u_{[k]}, \sqrt{2}v_{[k]})$, where $u$ and $v$ are random orthogonal unit vectors in $\R^n$.
\item The distribution $U_2$, where a draw from $U_2$ is the pair $(\sqrt{1+\eps_1} u_{[k]},\sqrt{2-\eps_2}v_{[k]})$, where $u$ and $v$ are random orthogonal unit vectors in $\R^n$.
\end{enumerate}

We note that in this case, we see only \emph{one} sample; the number of strong queries is the number of coordinates we see.  Here we use $u_{[k]},v_{[k]}$ to denote the first $k$ coordinates of $u,v$, and we have scaled to remove a naturally occurring factor of $(\frac{1}{5})^{1/4}$.  The success probability in this problem is exactly characterized by the total variation distance (see Definition~\ref{def:vd}) between $U_1,U_2$.  We prove that $d_{TV}(U_1,U_2)\leq k\sqrt{\eps}$ when $k=O(n^{1-c})$ for any $c>0$. The proof idea is that the vectors $u_{[k]}$ and $v_{[k]}$ have very similar distribution as a Gaussian unit vector (i.e., $N(0,1/n)^k)$ when $k = O(n^{1-c})$ for a constant $c > 0$.

We will show that it is hard to distinguish between the following distributions:

\begin{enumerate}
\item The distribution $N_1$, where a draw from $N_1$ is the pair $(g_{[k]}, \sqrt{2}h_{[k]})$, where $g$ and $h$ are two Gaussian vectors distributed as $N(0,1/n)^n$.
\item The distribution $N_2$, where a draw from $N_2$ is the pair
$(\sqrt{1+\eps_1} g_{[k]},\sqrt{2-\eps_{2}}h_{[k]})$,  where $g$ and $h$ are two Gaussian vectors distributed as $N(0,1/n)^n$.
\end{enumerate}

We know that $u_{[k]}$ and $v_{[k]}$ can be viewed as the upper left $2 \times k$ submatrix of a random orthogonal matrix.  For this, we use a theorem due to Li et al.~\cite{LNW14}.

\begin{theorem} \label{thm:LNW} Given an $n$ dimensional random orthogonal matrix $U$ as well as a random Gaussian matrix, if we take an $r \times k$ submatrix $U_{r,k}$ of $U$ and an $r \times k$ submatrix $G_{r,k}$ of $G$. We have that
\[
        d_{TV}(G_{r,k},U_{r,k})=o(1)
\]
when $r k\leq n^{1-\Omega(1)}$.
\end{theorem}

We can apply Theorem~\ref{thm:LNW} to get that
$d_{TV}(U_1,N_1), d_{TV}(U_2,N_2)\leq o(1)$ when $k=O(n^{1-\Omega(1)})$.

If we think of the pairs of vectors in $N_1$ and $N_2$ as concatenated vectors of length $2k$, we can write the $KL$ divergence between $N_1$ and $N_2$ as

\[
d_{KL}(N_1,N_2)= d_{KL}\left(N(0,\Sigma_0), N(0,\Sigma_1)\right)
\]

where

\[
\Sigma_0 =
\left[\begin{matrix}
\mathrm{Id}_{k \times k} & 0_{k \times k} \\
0_{k \times k} & 2 \cdot \mathrm{Id}_{k \times k} \\
\end{matrix}\right]
\qquad \mathrm{\; and \;} \qquad
\Sigma_1 =
\left[\begin{matrix}
(1 + \eps_1) \cdot \mathrm{Id}_{k \times k} & 0_{k \times k} \\
0_{k \times k} & (2 - \eps_2) \cdot \mathrm{Id}_{k \times k} \\
\end{matrix}\right]
\]

for some $\eps_1,\eps_2>0$ being $O(\eps)$.  By Theorem~\ref{thm:klg}, we have that
\begin{eqnarray*}
&&d_{KL}(N_1,N_2)\\ &=&  k (1+\eps_1) + k(1-\eps_2/2)-2k -k \ln(1+\eps_1)- k\ln(1-\eps_2/2) \\&=& O(k(\eps^2_1+\eps_2^2))\\ &=&O(k\eps^2)
\end{eqnarray*}

By Theorem~\ref{thm:kl}, we have that $d_{TV}(N_1,N_2) \leq O(\eps\sqrt{k})$.  Using the triangle inequality of variation distance, overall we have that $d_{TV}(U_1,U_2) \leq d_{TV}(N_1,N_2) + o(1) = O(\eps\sqrt{k})$. Therefore, we need at least $k= \Omega(1/\eps^2)$ strong queries in order to distinguish $P_1$ from $P_2$ with probability $2/3$, completing the proof.
\end{proof}

\section{Every $(\eps,\delta)$-estimator requires $\Omega(\frac{1}{\eps^2}\log\frac{1}{\delta} )$-queries}\label{sec:approxlb}

Our main theorem in this section is the following.

\begin{theorem} \label{thm:approxlb}For $0< \delta, \eps<1$, any $(\eps,\delta)$-estimator will require
  $\Omega(\frac{1}{\eps^2}\log(\frac{1}{\delta}))$ queries. This is
  true even if the unknown matrix $A$ has rank $1$.
    \label{th:nonadaptive}
\end{theorem}
We will assume that $\eps<1/3$.  As in the previous section, we reduce this estimation problem to the problem of distinguishing two distributions.   If we directly use the same problem defined in section~\ref{sec:varlb}, we can only get an $\Omega(\frac{1}{\eps})$ lower bound, and the lower bound comes from distributions supported on matrices with rank at least $2$. 
Here we consider a different problem for which we can also get the optimal dependence on $\eps$ and $\delta$, and uses distributions over rank $1$ matrices.  We will show that the following decision problem is hard.

\begin{definition}~\knote{Not quite a definition I guess...}
Given a uniformly random unit vector $u$, we want to
  distinguish the following two distributions:
  \begin{enumerate}
  \item The distribution $P_1$, where a draw from $P_1$ is the matrix $A_1 = u u^T$.
  \item The distribution $P_2$, where a draw from $P_2$ is the matrix $A_2= (1+3\eps)uu^T$.
  \end{enumerate}
The distinguisher has query access to an oracle that, on input $x \in \R^n$ and unknown matrix $A$ (from $P_1$ or $P_2$), returns $(\trace(A))u\cdot x$.  
  \end{definition}

The following lemma suggests that the above problem is easier than the trace estimating problem.
\begin{lemma} If we have an $(\eps,\delta)$-estimator that makes $k$ queries, then we can use it to distinguish the above two cases with accuracy at least $1-\delta$ using $k$ queries.
\end{lemma}
\begin{proof}
When the matrix has trace $1+3\eps$, an $(\eps,\delta)$-estimator will output a value above $(1+3\eps)(1-\eps)=1+2\eps-3\eps^2 > 1+\eps$ with probability at least $1-\delta$. On the other hand, when the matrix has trace $1$, the estimator will output some value below $1+\eps$ with probability at least $1-\delta$. Therefore, we output the distribution $P_1$ when the estimator's output is below $1+\eps$ and $P_2$ otherwise, this allows us to distinguish $P_1$ from $P_2$ with success probability at least $1-\delta$.
\end{proof}

We proceed to prove an $\Omega(\frac{1}{\eps^2}\log \frac{1}{\delta})$ lower bound for distinguishing $P_1$ from $P_2$ with accuracy $1-\delta$, which leads to essentially the same lower bound for the query complexity of $(\eps,\delta)$-estimator for the trace.

  \begin{lemma}
  \label{thm:len}
  If we make less than
  $\Omega\left(\frac{1}{\eps^2}\log (\frac{1}{\delta})\right)$ queries, we can not
  distinguish the above two cases with probability at least $1-\delta$.
\end{lemma}
\begin{proof}Notice that above problem is a special case of the problem defined in Definition~\ref{def:general}.  In the notation there, we want to distinguish $P_{1,0}$ and $P_{1+3\eps,0}$.  Therefore, we can apply Lemma~\ref{lem:strong}, and assumes that the $k$ queries are a set of random orthogonal unit vectors. Similar to the argument in Lemma~\ref{lem:testlb}, the above problem is equivalent to distinguish the following two distributions $U_1,U_2$ with one sample. 
 \begin{enumerate}
\item $U_1$: $u_{[k]}$.
\item $U_2$:$ (1+3\eps)u_{[k]}$.
\end{enumerate}
Here $u$ is a random unit vector, and $u_{[k]}$ is the first $k$ coordinates of $u$.

Again, we define the corresponding Gaussian distribution problem as follows.
\begin{enumerate}
\item $N_1$:  $g_{[k]}$.
\item $N_2$:  $ (1+3\eps )g_{[k]}$.
\end{enumerate}
for $g$ being a random vector whose entries are independent and distributed as $N(0,1/n)$.
We know that $d_{TV}(U_1,N_1)= d_{TV}(U_2,N_2)$. Here $U_1$ is distributed as the first $k$ coordinates of a random unit vector and $N_1$ is  $N(0,1/n)^k$.

We apply the following bound due to Khoklov~\cite{Kho06}.
\begin{lemma} $d_{TV}(U_1,N_1) = O(k/n)$.
\end{lemma}

Specifically, when $k=o(n)$, the distance between $U_i$ and $N_i$ is $o(1)$.  It remains to analyze $d_{TV}(N_1,N_2)$. If we use the same proof as in Lemma~\ref{lem:testlb} by computing the KL divergence, we would only get that $d_{TV}(N_1,N_2)\leq O(k\sqrt{\eps})$ which implies that $k=\Omega(1/\eps^2)$, completely independent of $\delta$.
 
In order to involve $\delta$ in the lower bound, we prove the following theorem that might be of independent interest.
\begin{theorem}
\label{thm:tv-normal}
Let $P_1$ be $N(0,1)^k$ and $P_2$ be $N(0,1+\theta)^k$, for any $0< \delta\leq 1$ and $0\leq \theta < C_0$ for $C_0$ being any positive constant (such as $C_0=100$). For the inequality
$d_{TV}(P_1,P_2)\geq 1- \delta$ to hold, it is necessary that $k\geq \Omega(\frac{1}{\theta^2}\log\frac{1}{\delta}$).
\end{theorem}

Assuming the correctness of above theorem,  we have the proof of Theorem~\ref{thm:approxlb} by setting $\theta=3\eps$. It is easy to see that the assumption that $\eps < 1/3$ here can be replaced by any constant $\eps\leq c$ for any $c$ bounded away from $1$ by essentially the same proof with $A_2= (1+m\eps)uu^T$  for any $m$ such that $m=2/(1-c)$.
\end{proof}

\begin{proof}We know that the probability density functions of $P_1$ and $P_2$ are $f_1(z) = (\frac{1}{\sqrt{2\pi}})^k e^{-\frac{\|z\|_2^2}{2}}$ and $f_2(z) = (\frac{1}{\sqrt{2\pi(1+\theta)}})^k e^{-\frac{\|z\|_2^2}{2(1+\theta)}}$, respectively.
\ignore{We can write their total variation distance as 
\begin{eqnarray*}
&d_{TV}(P_1,P_2 ) &= \frac{1}{2}\int_{z} |f_1(z) - f_2(z)| \; \mathrm{d}z \\
\ignore{=\frac{1}{2}\int_{z} |1- \frac{f_2(z)}{f_1(z)}|f_1(z) \; \mathrm{d}z} 
\end{eqnarray*}}
Suppose that for all $z$ in some set $S$, we have $f_2(z)/f_1(z)\in (\alpha,1/\alpha)$ for some constant $\alpha < 1$.  Without loss of generality, assuming $f_2(z) \leq f_1(z)$, we have
\[
\frac{|f_1(z)-f_2(z)|}{f_1(z)+f_2(z)} = \frac{1-f_2(z)/f_1(z)}{1+f_2(z)/f_1(z)} \leq
\frac{1-\alpha}{1+\alpha}
\]
and thus
\[
|f_1(z)-f_2(z)| \leq \frac{1-\alpha}{1+\alpha} (f_1(z) + f_2(z)).
\]
It follows that
\[
\ignore{d_{TV}(P_1,P_2)}
\int_{z:z \in S} |f_1(z) - f_2(z)| \; \mathrm{d}z 
\leq \frac{1-\alpha}{1+\alpha} \cdot \int_{z:z \in S} (f_1(z)+f_2(z)) \; \mathrm{d}z
\ignore{\leq \frac{1-\alpha}{1+\alpha}}.
\]

Of course, it is impossible to find a constant $\alpha$ when $z$ can be arbitrary.  We define $S =\{z \ |\ \|z\|^2_2\in (k-c, k+c), z\in \R^n\}$ for some parameter $c$ which we will specify later.
Let us also denote  $r_S=\min_{z\in S}\min\left(\frac{f_1(z)}{f_2(z)},\frac{f_2(z)}{f_1(z)}\right)$; i.e., the minimum ratio between $f_1(z)$ and $f_2(z)$ over $z\in S$. 

We then have 
\begin{eqnarray*}
&&d_{TV}(P_1,P_2 )\\&=& \frac{1}{2}\int_{z} |f_1(z) - f_2(z)|\mathrm{d}z \\&\leq& \frac{1}{2} \left(\Pr_{z\sim P_1}(z \notin S)+\Pr_{z\sim P_2}(z \notin S) +\int_{z:z \in S} |f_1(z) - f_2(z)|\mathrm{d}z \right)
\\ &\leq&  \frac{1}{2}\left(\Pr_{z\sim P_1}(z \notin S)+\Pr_{z\sim P_2}(z \notin S) + \frac{1-r_S}{1+r_S} \left(\Pr_{z\sim P_1}(z \in S) +\Pr_{z\sim P_2}(z \in S)\right)\right) 
\\&=&1-\frac{r_S}{1+r_S}\cdot \left(\Pr_{z\sim
P_1}(z \in S) +\Pr_{z\sim P_2}(z \in S)\right)
\end{eqnarray*}

In order to have $d_{TV}(P_1,P_2)\geq 1-\delta$, we must have that 
\[
\frac{r_S}{1+r_S} \left(\Pr_{z\sim
P_1}(z \in S) +\Pr_{z\sim P_2}(z \in S)\right)\leq \delta
\]

which we can weaken to

\[
r_S \left(\Pr_{z\sim
P_1}(z \in S) +\Pr_{z\sim P_2}(z \in S)\right)\leq \delta
\]

since $r_S \ge 0$.  For $z \in S$, we know that


\begin{eqnarray*}
&& \frac{f_1(z)}{f_2(z)}
\\ &=& (1+\theta)^{\frac{k}{2}} e^{-\frac{\theta\|z\|_2^2}{2(1+\theta)}}\\
&=&
\exp \left(\frac{k}{2}\ln(1+\theta)-\frac{\theta\|z\|_2^2}{2(1+\theta)}\right) \\
&\in&
\left(\exp \left( \frac{k}{2}\left(\ln(1+\theta)-\frac{\theta}{1+\theta}\right) - \frac{\theta c}{2(1+\theta)}\right), \right. \\ && 
 \left. \exp\left( \frac{k}{2}\left(\ln(1+\theta)-\frac{\theta}{1+\theta}\right) + \frac{\theta c}{2(1+\theta)}\right)\right)
\end{eqnarray*}

Defining $h(\theta) = \ln(1 + \theta) - \frac{\theta}{1 + \theta}$, we have $h'(\theta) = \frac{\theta}{(1 + \theta)^2}$.  For $0 < \theta < C_0$, we have $\frac{\theta}{(C_0+1)^2} \leq h'(\theta) \leq \theta$, so $\frac{\theta^2}{2(C_0+1)^2} \leq h(\theta) \leq \frac{\theta^2}{2}$ on this interval, and we have

\[
\frac{f_1(z)}{f_2(z)}
\in
\left(\exp \left( \frac{k\theta^2}{4(C_0+1)^2} - \frac{\theta c}{2(1+\theta)}\right),
\exp\left( \frac{k\theta^2}{4} + \frac{\theta c}{2(1+\theta)} \right) \right)
\]

\[
\frac{f_1(z)}{f_2(z)}
\in
\left(\exp \left( \frac{k\theta^2}{4(C_0+1)^2} - \frac{\theta c}{2}\right),
\exp\left( \frac{k\theta^2}{4} + \frac{\theta c}{2} \right) \right)
\]

So we can take
\begin{eqnarray*} &&r_S \\ &=& \min \left\{ \exp\left(-\dfrac{k\theta^2}{4} - \dfrac{\theta c}{2}\right), \exp\left(\dfrac{k\theta^2}{4(C_0+1)^2} - \dfrac{\theta c}{2}\right) \right\} \\&=& \exp\left(-\dfrac{k\theta^2}{4} - \dfrac{\theta c}{2}\right).\end{eqnarray*} The distribution on $\|z\|_2^2$ is a $\chi$-square distribution; we use the following tail estimate.

%
%


\begin{theorem}[Tail of $\chi$-square distribution]~\cite{LM00} \label{thm:chi}Let $X\sim \chi_k^2$, then 
\begin{itemize}
\item $\Pr(X>k +2\sqrt{k} t+2t^2 )\leq e^{-t^2}$
\item $\Pr(X<k+2\sqrt{k} t)\leq e^{-t^2}$.
\end{itemize}
\end{theorem}

We now set $c=4\sqrt{k}$, so that we have $\Pr_{z\sim
P_1}(z \in S) = \Omega(1)$ and  $\Pr_{z\sim
P_2}(z \in S) = \Omega(1)$.  This implies that $\exp \left(-\dfrac{k\theta^2}{4} - 2\sqrt{k}\theta\right) \leq C\delta$ for an absolute constant $C$, which in turn implies that $k$ should be $\Omega((1/\theta^2)\log (1/\delta))$. Since for any pair of distributions $N(0, \sigma_1^2)^k$ and $N(0, \sigma_2^2)^k$, after applying a factor $1/\sqrt{n}$ on both $\sigma_1$ and $\sigma_2$, the total variation distance doesn't change,  we complete the proof.

\end{proof}


\section*{Acknowledgement}
The second author is grateful for Yi Li and Siu-On Chan for helpful discussions.

\bibliographystyle{alpha}
\bibliography{everything}


%

\end{document}